\documentclass{article}

\usepackage[affil-it]{authblk}
\usepackage[usenames,dvipsnames]{xcolor}
\usepackage{amsfonts}
\usepackage{amsmath,amsthm,amssymb,dsfont}
\usepackage{enumerate}
\usepackage[english]{babel}
\usepackage{graphicx}	
\usepackage[caption=false]{subfig}
\usepackage[margin=3cm]{geometry}
\usepackage{url}
\usepackage{todonotes}
\usepackage{bbm}
\usepackage{booktabs}

\usepackage{tikz}
\usetikzlibrary{chains}
\usetikzlibrary{fit}
\usepackage{pgflibraryarrows}		%optional
\usepackage{pgflibrarysnakes}		%optional

\usepackage{makecell}

\usepackage{epsfig}
\usetikzlibrary{shapes.symbols,patterns} % for source symbols
\usepackage{pgfplots}
\pgfplotsset{compat=newest}

\usepackage{mathrsfs}

\usepackage{hyperref}
\hypersetup{colorlinks=true,citecolor=blue,linkcolor=blue,filecolor=blue,urlcolor=blue,breaklinks=true}

\usepackage{nicefrac}
\usepackage{mathtools}

\theoremstyle{plain}
\newtheorem{theorem}{Theorem}[section]

\newtheorem{proposition}[theorem]{Proposition}

\theoremstyle{definition}
\newtheorem{definition}[theorem]{Definition}

\newcommand*{\cH}{\mathcal{H}}

\newcommand*{\cN}{\mathcal{N}}

\newcommand*{\cR}{\mathcal{R}}

\newcommand*{\RR}{\mathbb{R}}

\newcommand*{\CC}{\mathbb{C}}

\newcommand*{\id}{I}

\newcommand*{\tr}{\mathrm{tr}}
\newcommand*{\ket}[1]{| #1 \rangle}

\newcommand*{\spr}[2]{\langle #1 | #2 \rangle}
\newcommand{\ketbra}[2]{|#1\rangle\!\langle #2|}

\newcommand{\sq}{E_{\mathrm{sq}}}

\definecolor{mycolor1}{rgb}{0.00000,0.44700,0.74100}%
\definecolor{mycolor2}{rgb}{0.85000,0.32500,0.09800}%
\definecolor{mycolor3}{rgb}{0.92900,0.69400,0.12500}%
\definecolor{mycolor4}{rgb}{0.49400,0.18400,0.55600}%

\newcommand{\bP}{\mathbf{P}}

\newcommand{\bip}{\boldsymbol{p}}
\newcommand{\bY}{{\bf Y}}
\newcommand{\bZ}{{\bf Z}}

\begin{document}

\title{{\LARGE Semidefinite programming lower bounds \\
on the squashed entanglement}}

\author[1]{Hamza Fawzi}
\author[2]{Omar Fawzi}

\affil[1]{\small{DAMTP, University of Cambridge, United Kingdom}}
\affil[2]{\small{Univ Lyon, ENS Lyon, UCBL, CNRS, Inria, LIP, F-69342, Lyon Cedex 07, France}}

\maketitle

%\tableofcontents

\begin{abstract}
The squashed entanglement is a widely used entanglement measure that has many desirable properties. However, as it is based on an optimization over extensions of arbitrary dimension, one drawback of this measure is the lack of good algorithms to compute it. Here, we introduce a hierarchy of semidefinite programming lower bounds on the squashed entanglement. 
%The semidefinite programs we define can also be seen as constructing entanglement witnesses that are different from more traditional ones.
\end{abstract}

%%%%%%%%%%%%%%%%%%%%%%%%%%%%%%%

\section{Introduction}
 
The squashed entanglement is an entanglement measure introduced by Christandl and Winter~\cite{CW04}. It satisfies many desirable properties of an entanglement measure such as monotonicity under local operations and classical communication (LOCC), additivity under tensor products, a monogamy relation~\cite{KW04} and faithfulness~\cite{BCY11}. In fact, it is the only known entanglement measure satisfying these properties. These mathematical properties also have operational applications: the squashed entanglement provides an upper bound on the distallable entanglement and the distillable key~\cite{Christandl06,christandl2007unifying,CSW12,wilde2016squashed}. For a bipartite state $\rho_{AB}$ on $A \otimes B$, the squashed entanglement is defined as
\begin{align}
\label{eq:squashed_def}
\sq(A:B)_{\rho} = \frac{1}{2} \inf_{\rho_{ABE}} I(A:B|E)_{\rho} \ ,
\end{align}
where the infimum is taken over all finite-dimensional Hilbert spaces $E$ and all extensions $\rho_{ABE}$ of the state $\rho_{AB}$, i.e., $\tr_{E} \rho_{ABE} = \rho_{AB}$. The conditional mutual information is defined in equation~\eqref{eq:cmi} below.
Note that as shown in~\cite{shirokov2016squashed}, considering infinite dimensional separable Hilbert spaces $E$ leads to the same quantity. We also note that a similar quantity was defined in~\cite{tucci1999quantum,tucci2002entanglement}.

To define the conditional mutual information it is useful to start with the quantum relative entropy. For density operators $\rho, \sigma$ on a finite-dimensional Hilbert space $\cH$, the quantum relative entropy is defined by
\begin{align}
\label{eq:relative_entropy}
D(\rho \| \sigma) = 
\left\{\begin{array}{ll}
        \tr(\rho (\log \rho - \log \sigma)) & \text{if } \mathrm{supp}(\rho) \subseteq \mathrm{supp}(\sigma) \\
        + \infty & \text{otherwise.}
        \end{array}\right. 
\end{align}
For a density operator $\rho_{A E}$ on $A \otimes E$, we define the conditional von Neumann entropy $H(A|E)_{\rho} = -D(\rho_{AE} \| \id_{A} \otimes \rho_E)$ where $\rho_{E} = \tr_{A} \rho_{AE}$ and $\id_{A}$ denotes the identity operator on $A$. The conditional mutual information is then defined by
\begin{equation}
\label{eq:cmi} 
I(A:B|E)_{\rho} = H(A|E)_{\rho} - H(A|BE)_{\rho} \ .
\end{equation}
Note that by considering a purification $\rho_{ABDE}$ of the state $\rho_{ABE}$ and using the fact that $H(A|BE)_{\rho} = - H(A|D)_{\rho}$, we can write the squashed entanglement as follows:
\begin{align}
\label{eq:squashed_purification}
\sq(A:B)_{\rho} = \frac{1}{2} \inf_{\rho_{ABDE}} H(A|D)_{\rho} + H(A|E)_{\rho} \ ,
\end{align}
where the infimum is over arbitrary finite-dimensional Hilbert spaces $D$ and $E$ and extensions $\rho_{ABDE}$ of the state $\rho_{AB}$.

\section{Squashed entanglement and $f$-divergences}

To obtain our hierarchy of semidefinite programs, it is useful to generalize the squashed entanglement~\eqref{eq:squashed_def} to standard $f$-divergences (also known as quasi-relative entropies). Let $f : (0,\infty) \to \RR$ be a continuous function 
such that $\lim_{x \to \infty} \frac{f(x)}{x} = 0$ and let $f(0^+) := \lim_{x \to 0} f(x)$ (which could be infinite).
\begin{definition}[Standard $f$-divergences~\cite{Pet86,HM17}]
\label{def:standard_f_divergence}
Let $\rho, \sigma$ be positive semidefinite operators on $\cH$ with spectral decompositions $\rho = \sum_{j} \lambda_j P_j$ and $\sigma = \sum_{k} \mu_k Q_k$. We define
\begin{align}
\label{eq:def_standard_f_divergence}
D_{f}(\rho \| \sigma) := \sum_{j : \lambda_j > 0} \sum_{k : \mu_k > 0} \lambda_j f(\mu_k \lambda_j^{-1} ) \tr(P_j Q_k) + f(0^+) \tr \left( \rho (I-\sigma^0) \right) \ ,
\end{align}
where $\sigma^0$ denotes the projector onto the support of $\sigma$.\footnote{Note that in this definition, the roles of $\rho$ and $\sigma$ are sometimes inverted. }

If $f$ is operator convex, then $D_{f}$ is jointly convex and satisfies the data processing inequality, i.e., for any quantum channel $\cN$, we have $D_{f}(\cN(\rho) \| \cN(\sigma)) \leq D_{f}(\rho \| \sigma)$.
\end{definition}
Note that $D_{f}(\rho \| \sigma)$ is linear and monotone in $f$ for fixed positive semidefinite operators $\rho, \sigma$, i.e., $D_{f}(\rho \| \sigma) \leq D_{g}(\rho \| \sigma)$ if $f(x) \leq g(x)$ for all $x \in (0,\infty)$.
When $f(x) = -\log x$, we obtain the quantum relative entropy.  As done in~\cite{FSP17,BFF21}, we use an approximation of the function $-\log$ using rational functions based on a Gauss-Radau quadrature (see e.g.,~\cite[page 103]{davis1984methods}). In order to do that, we introduce the functions $f_t : (0,+\infty) \to \RR$ for $t \in [0,1]$ as
\begin{align}
\label{eq:def_ft}
f_t(x) := \frac{x-1}{t (x-1) + 1} \ .
\end{align}
It is simple to see that $-f_t$ is operator convex for all $t \in [0,1]$. In addition, the condition $\lim_{x \to \infty} \frac{-f_t(x)}{x} = 0$ is satisfied for all $t \in (0,1]$, but not when $t = 0$. For $t = 0$, we take the convention that $D_{-f_0}$ is still defined via~\eqref{eq:def_standard_f_divergence}, i.e., $D_{-f_0}(\rho \| \sigma) = 1 - \tr(\rho^0 \sigma)$ (which is closely related to the so-called min-relative entropy $D_{\min}(\rho \| \sigma) = - \log(\rho^0 \sigma)$), and not as in~\cite{HM17} which would simply be $1 - \tr(\sigma)$. 
%Note that with our convention, the statement of Lemma~\ref{lem:dualitypure} holds for the functions $f_t$ for the whole range $t \in [0,1]$. 
With our convention $D_{-f_0}(\rho \| \sigma) = \lim_{t \to 0} D_{-f_t}(\rho \| \sigma)$. In summary, for all $t \in [0,1]$, $D_{-f_t}$ is jointly convex and is thus a valid divergence satisfying the data-processing inequality.

The following proposition shows that one can choose $t_1, \dots, t_m$ and positive weights $w_1, \dots, w_m$ such that taking the linear combination of these divergences gives a good approximation of the quantum relative entropy.
\begin{proposition}[Gauss-Radau quadrature]
For any positive integer $m \geq 1$, there exists a choice of nodes $t_1, \dots, t_m \in (0,1]$ and weights $w_1, \dots, w_m > 0$ such that for all $x \in (0, + \infty)$
\begin{align}
\label{eq:log_rational}
0 \leq \left(-\frac{1}{\ln 2} r_m(x) \right) - \left( - \log x \right) \leq \frac{1}{m^2 \ln 2}\left(x + \frac{1}{x} - 2 \right)  \quad \text{ where } \quad r_m(x) := \sum_{i=1}^m w_i f_{t_i}(x) \ .
\end{align}
\end{proposition}
\begin{proof}
As in~\cite{BFF21}, we fix the node $t_m = 1$ and the Gauss-Radau quadrature gives the existence of nodes $t_1, \dots, t_{m-1} \in (0,1)$ and weights $w_1, \dots, w_m > 0$ such that $\ln x - r_m(x) \geq 0$ for all $x \in (0, +\infty)$. 
%In addition $w_m = \frac{1}{m^2}$. 
This establishes the first inequality $0 \leq \left(-\frac{1}{\ln 2} r_m(x) \right) - \left( - \log x \right)$.

To prove the second inequality, we will establish the statement
\begin{align}
\label{eq:relation_r_barr}
 r_m(x) + r_m(1/x)  \geq \frac{-1}{m^2}\left(x + \frac{1}{x} - 2\right) \ .
\end{align}
Note using the fact that $\ln x - r_m(x) \geq 0$, the inequality~\eqref{eq:relation_r_barr} implies that 
\begin{align*}
r_m(x) - \ln(x) \geq r_m(x) - \ln x - (\ln(1/x) - r_m(1/x)) \geq \frac{-1}{m^2}\left(x + \frac{1}{x} - 2\right) \ ,
\end{align*}
which establishes the desired upper bound in~\eqref{eq:log_rational}.

We now focus on~\eqref{eq:relation_r_barr}. Using the Gauss-Radau quadrature property, one can show that we have $\ln(x) - r_m(x) = O( (x-1)^{2m})$ in the neighborhood of $x = 1$ and $w_m = \frac{1}{m^2}$ (see~\cite{BFF21} for more details). Recalling that
\begin{align*}
r_m(x) 
&= \sum_{i=1}^{m-1} w_i \frac{x-1}{t_i(x-1) + 1} + \frac{1}{m^2} \frac{x-1}{x} \\
&= \sum_{i=1}^{m-1} \frac{w_i}{1-t_i} \frac{x-1}{\frac{t_i}{1-t_i} x + 1} + \frac{1}{m^2} \frac{x-1}{x}
\end{align*}
and putting all the fractions to the same denominator, we can write $r_m(x) = \frac{-1}{m^2} \frac{F_m(x)}{x G_{m-1}(x)}$ where 
\begin{align*}
G_{m-1}(x) &= \prod_{i=1}^{m-1} \left( \frac{t_i}{1-t_i} x + 1 \right) \\
- F_m(x) &= \sum_{i=1}^{m-1} \frac{w_i}{1-t_i} (x - 1) \cdot x \cdot \prod_{j \in \{1, \dots, m-1\}, j \neq i} \left(\frac{t_j}{1-t_j} x + 1\right) + (x-1) \cdot \prod_{j \in \{1, \dots, m-1\}} \left(\frac{t_j}{1-t_j} x + 1\right) \ .
\end{align*}
Note that $\deg F_m = m$, $\deg G_{m-1} = m-1$ and they are normalized so that $F_m(0) = G_{m-1}(0) = 1$. 
%The polynomial $G_{m-1}$ is given by the product of the denominators in \eqref{eq:rm111}, appropriately normalized:
%\[
%G_{m-1} = \prod_{i=1}^{m-1} \left( \frac{t_i}{1-t_i} x + 1 \right).
%\]
Then we have:
\[
r_m(1/x) = -\frac{1}{m^2} \frac{xF_m(1/x)}{G_{m-1}(1/x)} = -\frac{1}{m^2} \frac{\hat{F}_m(x)}{\hat{G}_{m-1}(x)}
\]
where we let $\hat{F}_m(x) = x^m F_m(1/x)$ and $\hat{G}_{m-1} = x^{m-1} G_{m-1}(1/x)$. Note that $\hat{F}_m$ and $\hat{G}_{m-1}$ are both monic polynomials (i.e., the leading coefficient is $1$) since $F_m(0) = G_{m-1}(0) = 1$.

Summing we get
\begin{equation}
\label{eq:rm+rm1x}
r_m(x) + r_m(1/x) = -\frac{1}{m^2} \frac{F_m(x) \hat{G}_{m-1}(x) + x\hat{F}_m(x) G_{m-1}(x)}{x G_{m-1}(x) \hat{G}_{m-1}(x)}.
\end{equation}

We know that around $x=1$, $r_m(x) - \ln(x) = O((x-1)^{2m})$, which implies $r_m(1/x) + \ln(x) = O((x-1)^{2m})$ and so by summing we get $r_m(x) + r_m(1/x) = O((x-1)^{2m})$. Since the numerator of \eqref{eq:rm+rm1x} has degree $2m$ this means that it has to be a multiple of $(x-1)^{2m}$, i.e., we must have
\[
r_m(x) + r_m(1/x) = -\frac{1}{m^2} \frac{c_0 (x-1)^{2m}}{x G_{m-1}(x) \hat{G}_{m-1}(x)}
\]
where $c_0$ is the leading term of the numerator, i.e., it is the leading term of $G_{m-1}$ so $c_0 = \prod_{i=1}^{m-1} \frac{t_i}{1-t_i}$.
% (numerically this seems to be equal to $1/m$ but we will not need this.)
The main observation now is to note that $G_{m-1}(x) \hat{G}_{m-1}(x) \geq (x+1)^{2m-2}$. Indeed, we have
\[
\begin{aligned}
G_{m-1}(x) \hat{G}_{m-1}(x) &= \prod_{i=1}^{m-1} \left(\frac{t_i}{1-t_i} x + 1 \right) \left(\frac{t_i}{1-t_i} +  x\right)\\
&= c_0 \prod_{i=1}^{m-1} \left(x + \frac{1-t_i}{t_i} \right) \left(\frac{t_i}{1-t_i} +  x\right)\\
&\geq c_0 \prod_{i=1}^{m-1} (x+1)^2 = c_0 (x+1)^{2(m-1)}
\end{aligned}
\]
where we used the fact that for $x,a > 0$, $(x+a)(x+1/a) \geq (x+1)^2$.
This implies that
\begin{align*}
r_m(x) + r_m(1/x) 
&\geq -\frac{1}{m^2} \frac{(x-1)^2}{x} \left(\frac{x-1}{x+1}\right)^{2m-2} \\
&\geq -\frac{1}{m^2} \frac{(x-1)^2}{x} 
\end{align*}
as desired.
\end{proof}

In analogy with the expression~\eqref{eq:squashed_purification}, we can define a variant of the squashed entanglement for $f$-divergences.
\begin{definition}[$f$-squashed entanglement]
\label{def:f-squashed_ent}
 Let $\rho_{AB}$ be a bipartite state acting on $A \otimes B$. We define 
\begin{align*}
\sq^{(f)}(A:B)_{\rho} := \frac{1}{2} \inf_{\rho_{ABDE}} - D_{f}(\rho_{AD} \| \id_{A} \otimes \rho_D) - D_{f}( \rho_{AE} \| \id_{A} \otimes \rho_E) \ ,
\end{align*}
where the infimum is over all extensions $\rho_{ABDE}$ of $\rho_{AB}$. In particular, we define
\begin{align*}
\sq^{(m)}(A:B)_{\rho} := \sq^{(-\frac{r_m}{\ln 2})}(A:B)_{\rho} \quad \text{ where  $r_m$ is defined in~\eqref{eq:log_rational}} .
\end{align*}
%when $f(x) = -\frac{1}{\ln 2} \sum_{i=1}^m w_i f_{t_i}(x)$ from the right-hand side of~\eqref{eq:log_rational}, we use for short $\sq^{(m)}(A:B)_{\rho}$ for $\sq^{(f)}(A:B)_{\rho}$. 
Note that we have
\begin{align*}
\sq(A:B)_{\rho} = \sq^{(-\log)}(A:B)_{\rho} \geq \sq^{(m)}(A:B)_{\rho} \qquad \text{ for any } m \geq 1 \ .
\end{align*}
\end{definition}
We emphasize that in general the $f$-squashed entanglement could be negative and is not an entanglement measure. It will however be a useful tool in the rest of the paper. The main property of the $f$-divergences for $f = -\frac{r_m}{\ln 2}$ that we will use to construct semidefinite programs is a variational expression of $D_{f}$ that we introduce in the next section. Before that we establish basic properties of $\sq^{(m)}(A:B)_{\rho}$.

\begin{proposition}
\label{prop:sqm_sq}
For any state $\rho_{AB}$ and positive integer $m$, we have
\begin{align*}
\sq^{(m)}(A:B)_{\rho} \leq \sq(A:B)_{\rho} \leq  \sq^{(m)}(A:B)_{\rho} + \frac{2 d_{A} - 2}{m^2 \ln 2} \ .
\end{align*}
In addition, if $\rho_{AB}$ is pure, we have
\begin{align*}
\sq^{(m)}(A:B)_{\rho} &= \frac{1}{\ln 2} \tr\big(\rho_{A} r_m( \rho_{A}^{-1} ) \big) \ ,
\end{align*}
where $\rho_{A}^{-1}$ denotes the generalized inverse of $\rho_{A} = \tr_{B} \rho_{AB}$.
\end{proposition}
\begin{proof}
We simply use Definition~\ref{def:standard_f_divergence} together with inequality~\eqref{eq:log_rational}. This gives
\begin{align*}
\sq(A:B)_{\rho} 
&= \frac{1}{2} \inf_{\rho_{ABDE}} - D(\rho_{AD} \| \id_{A} \otimes \rho_D) - D(\rho_{AE} \| \id_{A} \otimes \rho_E) \\
&\leq \frac{1}{2} \inf_{\rho_{ABDE}} - D_{g}(\rho_{AD} \| \id_{A} \otimes \rho_D) - D_{g}(\rho_{AE} \| \id_{A} \otimes \rho_E) \ ,
\end{align*}
where $g(x) = -\frac{1}{\ln 2} r_m(x) - \frac{1}{m^2 \ln 2} \left(x + \frac{1}{x} - 2 \right)$. Note that for the function $f_0(x) = x - 1$ and $f_1(x) = (x - 1)/x$, we have
\begin{align*}
D_{f_0}(\rho_{AD} \| \id_{A} \otimes \rho_D) = \tr(\rho_{AD}^0 \id_{A} \otimes \rho_D) - 1 \leq d_{A} - 1 \\
D_{f_1}(\rho_{AD} \| \id_{A} \otimes \rho_D) = 1 - \tr(\rho_{AD}^2 \id_{A} \otimes \rho_{D}^{-1})  \geq 1 - d_{A} \ .
\end{align*}
%\red{Hamza: is this where the convention on $f_0$ is used?} Yes
As a result,
\begin{align*}
- D_{g}(\rho_{AD} \| \id_{A} \otimes \rho_D) \leq - D_{-\frac{r_m}{\ln 2}}(\rho_{AD} \| \id_{A} \otimes \rho_D) + \frac{1}{m^2 \ln 2} 2 (d_A - 1) \ ,
\end{align*}
which proves the desired bound.
%It suffices to use inequality~\eqref{eq:log_rational} together with the Definition~\ref{def:standard_f_divergence} of standard $f$-divergences.

If $\rho_{AB}$ is pure, then any purification has the form $\rho_{ABDE} = \rho_{AB} \otimes \rho_{DE}$ and as a result,
\begin{align*}
D_{f}(\rho_{AD} \| \id_{A} \otimes \rho_{D}) &= D_{f}(\rho_{A} \| \id_{A}) \\
&= \tr \left( \rho f(\rho^{-1}) \right) \ ,
\end{align*}
using the defining formula~\eqref{eq:def_standard_f_divergence} and using the notation $\rho^{-1}$ for the generalized inverse of $\rho$.
\end{proof}

%For positive integers $m$, we introduce the function $r_m : (0,+\infty) \to \RR$:
%\begin{align*}
%r_m(x) = - \frac{1}{\ln 2} \left( \sum_{i=1}^m w_i \frac{x-1}{t_i (x-1) + 1} + w_m \frac{x-1}{x} \right) \ . 
%\end{align*}
%\footnote{Note that compared to $r_m$ is not exactly the same. Here we included the factor $- \frac{1}{\ln 2}$ into the definition of $r_m$.}
%\begin{align*}
%-\log(x) \leq r_m(x) \ .
%\end{align*}

\section{Approximating the squashed entanglement with noncommutative polynomial optimization}

We now use the variational expression for $D_{-f_t}$ that was established in~\cite{BFF21}.
%use a recently established upper bound on the quantum relative entropy established in~\cite{BFF21}. 
\begin{theorem}[\cite{BFF21}]
\label{thm:dft_var_expr}
Let $t \in [0,1]$ and recall that $f_t(x) = \frac{x-1}{t(x-1) + 1}$. Let $\rho$ and $\sigma$ be positive semidefinite operators on a finite-dimensional Hilbert space $\cH$. Then
\begin{equation}
	\label{eq:DubvN}
		D_{-f_t}(\rho\|\sigma) = 
		- \inf_{Z} \frac{1}{t}  \left\{\tr(\rho) + \tr(\rho (Z+Z^*)) + (1-t) \tr(\rho Z^* Z) + t \tr(\sigma Z Z^*) \right\} \ ,
\end{equation}
where the infimum is over bounded operators $Z$ on $\cH$.
\end{theorem}

We fix an orthonormal basis for $A$ that we denote $\{\ket{a}\}_{a \in [d_A]}$ and one for $B$ that we denote $\{\ket{b}\}_{b \in [d_B]}$. 
For a Hilbert space $\cH$, we say that the density operator $\rho_{AB\cH}$ acting on $A \otimes B \otimes \cH$ is an extension of $\rho_{AB}$ if $\tr_{\cH} \rho_{AB\cH} = \rho_{AB}$. We write $\rho_{AB\cH} = \sum_{a_1a_2b_1b_2} \ketbra{a_1}{a_2} \otimes \ketbra{b_1}{b_2} \otimes \rho[a_1b_1,a_2b_2]$ where $\rho[a_1b_1,a_2b_2]$ is a bounded linear operator on $\cH$. With this notation, an extension of $\rho_{AB}$ is defined by a family of operator $\rho[a_1b_1,a_2b_2]$ satisfying $\tr\rho[a_1b_1,a_2b_2] = \tr( \ketbra{a_1}{a_2} \otimes \ketbra{b_1}{b_2} \rho_{AB})$.
%\red{Hamza: The RHS we can write it as $\spr{a_1 b_1}{\rho_{AB}| a_2 b_2}$?} Yes
\begin{theorem}
\label{prop:main_lb_sq_ent}
Let $m \geq 1$ and let $t_1, \dots, t_m, w_1, \dots, w_m$ be the nodes and weights from Proposition~\ref{eq:log_rational}. We define the matrix-valued polynomial ${\bf P}^{(m)}$ in noncommuting variables ${\bf Z} = \{Z_i[a_1, a_2], Z_i[a_1, a_2]^*\}_{i, a_1, a_2}$ by
\begin{align}
{\bf P}^{(m)}({\bf Z}) := \left(\sum_{i=1}^m P^i_{a_1, a_2}({\bf Z}) \right)_{a_1, a_2 \in [d_A]}
\end{align}
where 
\begin{align}
P^i_{a_1, a_2}({\bf Z}) &= \frac{w_i}{t_i \ln 2} \Bigg( Z_i[a_1, a_2] + Z_i[a_2, a_1]^* + (1-t_i) \sum_{a_3} Z_i[a_3, a_1]^* Z_i[a_3, a_2] \notag \\
&\qquad + \delta_{a_1=a_2}\left(1 + t_i \sum_{a_3, a_4} Z_i[a_3, a_4] Z_i[a_3, a_4]^* \right) \Bigg) \label{eq:def-pi}  \ ,
\end{align}
where $\delta_{a_1 = a_2} = 0$ when $a_1 \neq a_2$ and $\delta_{a_1 = a_2} = 1$ otherwise. Note that when the variables $Z_i[a_1, a_2]$ are replaced with operators on $\cH$, then ${\bf P}^{(m)}({\bf Z})$ is a Hermitian operator of $A \otimes \cH$. 
Then for any density operator $\rho_{AB}$ on $A \otimes B$, we have
\begin{equation}
\label{eq:ncpoly_sq_m}
\sq^{(m)}(A:B)_{\rho} = \inf \tr\left(\rho_{A \cH} \frac{ {\bf P}^{(m)}( {\bf Z}) + {\bf P}^{(m)}( {\bf Y}) }{2} \right)
%\inf \sum_{a_1, a_2, b} \tr\left(\rho[a_1b,a_2b] \cdot \frac{1}{2} \sum_{i=1}^m P^i_{a_1, a_2}({\bf Z}_i, {\bf Z}_i^*) + P^i_{a_1, a_2}({\bf Y}_i, {\bf Y}_i^*) \right)
%\inf_{\{Z_{i}[k,\ell]\}_{i \in [m], k \in [d_A], \ell \in [d_A]} } \frac{1}{2} \sum_{i=1}^{m} \frac{w_i}{t_i \ln 2}  \left\{1 + \tr(\rho[k,\ell] (Z[k,\ell]+Z[k,\ell]^*)) + (1-t_i) \tr(\rho[k,\ell] Z^* Z) + t_i \tr(\sigma ZZ^*) \right\}
\end{equation}
where the infimum is over finite-dimensional Hilbert spaces $\cH$, extensions $\rho_{AB\cH}$ of the state $\rho_{AB}$ and two mutually commuting families of operators ${\bf Z}$ and ${\bf Y}$ on $\cH$.
\end{theorem}

\begin{proof}
We start from the expression in Definition~\eqref{def:f-squashed_ent} with $f = -\frac{r_m}{\ln 2}$ and consider a fixed extension $\rho_{ABDE}$.
% the state $\rho_{ABDE}$ can then be written in block form as follows $\rho_{ABDE} = \sum_{a_1,a_2,b_1,b_2} \ketbra{a_1}{a_2} \otimes \ketbra{b_1}{b_2} \otimes \rho[a_1b_1,a_2b_2]$, where $\rho[a_1b_1,a_2b_2]$ is a linear operator on $D \otimes E$. Note that the reduced state $\rho_{AD} = \sum_{a_1,a_2} \ketbra{a_1}{a_2} \otimes \sum_{b} \tr_{E} \rho[a_1b,a_2b]$. We will use $\rho[a_1,a_2]$ as a shorthand for $\sum_{b} \rho[a_1b,a_2b]$.

We have 
\begin{align*}
& D_{-\frac{r_m}{\ln 2}}(\rho_{AD} \| \id_{A} \otimes \rho_D) \\
&= \frac{1}{\ln 2} \sum_{i=1}^m w_i D_{-f_{t_i}}( \rho_{AD} \| \id_{A} \otimes \rho_D) \\
&= - \sum_{i=1}^m \frac{w_i}{t_i \ln 2} \inf_{Z_i} \tr(\rho_{AD}(I+Z_i+Z_i^*)) + \tr\left(\rho_{AD} \left(  (1-t_i) Z_i^*Z_i + t_i \id_{A} \otimes \tr_{A} Z_i Z_i^* \right) \right)  \ ,
\end{align*}
where $Z_i$ is a bounded operator on $A \otimes D$ and we used Theorem~\ref{thm:dft_var_expr}. We can write such operators in the form $Z_i = \sum_{a_1,a_2} \ketbra{a_1}{a_2} \otimes Z_i[a_1,a_2]$. With this notation, we have
\begin{align*}
Z_i^* &= \sum_{a_1,a_2} \ketbra{a_1}{a_2} \otimes Z_i[a_2,a_1]^* \\
Z_i^* Z_i &= \left(\sum_{a_1,a_2} \ketbra{a_1}{a_2} \otimes Z_i[a_2,a_1]^* \right) \left(\sum_{a_3,a_4} \ketbra{a_3}{a_4} \otimes Z_i[a_3,a_4]\right) = \sum_{a_1,a_4} \ketbra{a_1}{a_4} \otimes \sum_{a_2} Z_i[a_2, a_1]^* Z_i[a_2, a_4] \\
Z_i Z_i^* &= \sum_{a_1,a_4} \ketbra{a_1}{a_4} \otimes \sum_{a_2} Z_i[a_1, a_2] Z_i[a_4, a_2]^* \ .
\end{align*}
Writing $\rho_{AD} = \sum_{a_1,a_2} \ketbra{a_1}{a_2} \otimes \rho[a_1, a_2]$ we get
\begin{align*}
- D_{-\frac{r_m}{\ln 2}}(\rho_{AD} \| \id_{A} \otimes \rho_D)
&= \inf_{\bf Z} \sum_{a_1,a_2} \tr\left( \rho[a_1,a_2] \sum_{i = 1}^m P^i_{a_1,a_2}({\bf Z}) \right) \\
&= \inf_{\bf Z} \tr\left( \rho_{AD} {\bf P}^{(m)}({\bf Z}) \right) \ ,
% \left( Z_i[a;a'] + Z[a';a]^* + (1-t_i) \sum_{\alpha} Z_i[\alpha;a]^* Z_i[\alpha;a'] \right) \right) \\
%&+ \sum_{i \in [m]} \frac{w_i}{t_i \ln 2} \sum_{a,a',\alpha} t_i \tr \left(  \rho[a;a]  Z_i[a'; \alpha] Z_i[a'; \alpha]^*  \right)  %\sum_{\alpha} (1-t_i)  \tr(\rho[a;a'] Z_i[\alpha;a]^* Z_i[\alpha;a']) + t_i \tr( \sum_{a} \rho[a,a] \sum_{a,\alpha} Z_i[a;\alpha] Z^*[a;\alpha]^*) \big\}
\end{align*}
where the infimum is over all operators ${\bf Z} = \{Z_i[a_1,a_2], Z_i[a_1,a_2]^*\}_{i,a_1,a_2}$ acting on $D$ and the noncommutative polynomials $P^i_{a_1, a_2}({\bf Z})$ are defined by~\eqref{eq:def-pi}.
%\begin{align*}
%P^i_{a_1, a_2}(Z_i, Z_i^*) &= \frac{w_i}{t_i \ln 2} \Bigg( Z_i[a_1, a_2] + Z_i[a_2, a_1]^* + (1-t_i) \sum_{a_3} Z_i[a_3, a_1]^* Z_i[a_3, a_2] \\
%&\qquad + \delta_{a_1=a_2}\left(1 + t_i \sum_{a_3, a_4} Z_i[a_3, a_4] Z_i[a_3, a_4]^* \right) \Bigg)  \ .
%\end{align*}

Proceeding in the same way for the term $-D_{-\frac{r_m}{\ln 2}}(\rho_{AE} \| \id_{A} \otimes \rho_{E})$, we can write for a given extension $\rho_{ABDE}$:
\begin{align*}
-D_{-\frac{r_m}{\ln 2}}(\rho_{AD} \| \id_{A} \otimes \rho_{D}) -D_{-\frac{r_m}{\ln 2}}(\rho_{AE} \| \id_{A} \otimes \rho_{E}) &= \inf_{{\bf Z}, {\bf Y}} \tr( \rho_{ADE} ({\bf P}^{(m)}({\bf Z}) \otimes \id_{E} + {\bf P}^{(m)}({\bf Y})\otimes \id_{D} )) \ ,
\end{align*} 
where ${\bf Z}$ and ${\bf Y}$ are families of operators acting on $D$ and $E$ respectively. Taking the infimum over all extensions $\rho_{ABDE}$, we obtain 
\begin{align*}
\sq^{(m)}(A:B)_{\rho} = \inf \tr\left(\rho_{A D E} \frac{ {\bf P}^{(m)}( {\bf Z}) \otimes \id_{E} + {\bf P}^{(m)}( {\bf Y}) \otimes \id_{D} }{2} \right) \ ,
\end{align*} 
where the infimum is over all choices of finite-dimensional spaces $D$ and $E$ and operators ${\bf Z}$ and ${\bf Y}$ on $D$ and $E$ respectively. Note that the operators $Z_{i}[a_1,a_2] \otimes \id_{E}$ and $\id_{D} \otimes Y_{j}[a_3, a_4]$ act on the tensor product space $\cH = D \otimes E$ and commute. Thus, we clearly obtain a lower bound by taking the infimum over finite-dimensional Hilbert spaces $\cH$ and now both collections ${\bf Z}$ and ${\bf Y}$ act on $\cH$ and they mutually commute. This shows the inequality $\geq$ of equation~\eqref{eq:ncpoly_sq_m}. Equality then follows from the now well-known fact that two commuting families of operators on a finite-dimensional Hilbert space can be mapped isomorphically to operators acting on different tensor factors of a tensor product of two finite-dimensional Hilbert spaces, see~\cite{scholz2008tsirelson} for a proof.
%
%We can do exactly the same thing for $H(A|E)_{\rho}$. We get
%\begin{align*}
%H(A|E)_{\rho} 
%&\geq \inf \sum_{a_1,a_2} \tr\left( \rho[a_1,a_2] \sum_{i=1}^m P^i_{a_1,a_2}(Y_i,Y_i^*) \right) \ ,
%\end{align*}
%where the infimum is over all operators $\{Y_i[a_1, a_2]\}_{i,a_1,a_2}$ acting on $E$. Adding up the two expressions, we have
%\begin{align*}
%H(A|D)_{\rho} + H(A|E)_{\rho}
%&\geq \inf_{{\bf Z}} \sum_{a_1,a_2} \tr\left( \rho[a_1,a_2] \sum_{i=1}^m \left( P^i_{a_1,a_2}(Z_i,Z_i^*) + P^i_{a_1, a_2}(Y_i, Y_i^*) \right) \right) \ \\
%&
%\end{align*}
%where the operators $Z_i[a_1, a_2]$ act on $D$ and $Y_i[a_1,a_2]$ act on $E$. Note that by construction, the operators $Z_{i_1}[a_1, a_2]$ and $Y_{i_2}[a_3, a_4]$ commute for any $i_1, i_2, a_1, a_2, a_3, a_4$. We can now take the infimum over all extensions and get the following lower bound on the squashed entanglement
%\begin{equation}
%\label{eq:sq222}
%\sq(A:B)_{\rho} \geq \inf \frac{1}{2} \sum_{a_1, a_2, b} \tr\left( \rho[a_1 b, a_2 b] \sum_{i \in [m]} \left( P^i_{a_1,a_2}(Z_i,Z_i^*) + P^i_{a_1, a_2}(Y_i, Y_i^*) \right) \right) \ ,
%\end{equation}
%where the infimum is over all extensions $\rho_{AB\cH} = \sum_{a_1a_2b_1b_2} \ketbra{a_1}{a_2} \otimes \ketbra{b_1}{b_2} \otimes \rho[a_1 b, a_2 b]$ of $\rho_{AB}$ and the families of operators $\{Z_i[a_1, a_2]\}_{i, a_1, a_2}$ and $\{Y_i[a_1, a_2]\}_{i, a_1, a_2}$ both act on $\cH$ and mutually commute. 
\end{proof}

\subsection{Semidefinite relaxation} 

The expression in~\eqref{eq:ncpoly_sq_m} is a noncommutative polynomial optimization problem and thus one can use semidefinite programming hierarchies for it~\cite{PNA10}. We now explain how this can be done. Let $\cR$ be the ring of noncommutative polynomials in the $4md_A^2$ variables 
\[
\{\bY,\bZ\} = \{Y_i[a_1,a_2],Y_i[a_1,a_2]^*,Z_i[a_1,a_2],Z_i[a_1,a_2]^*, \;\; i \in [m], \;\; a_1,a_2 \in [d_A]\} \, ,
\]
modulo the commutation relations
\begin{equation}
\label{eq:comm}
\begin{aligned}
\Bigl[ Y_i[a_1,a_2] \, , \, Z_{i'}[a'_1,a'_2] \Bigr] = \Bigl[ Y_i[a_1,a_2] \, , \, Z_{i'}[a'_1,a'_2]^* \Bigr] = 0
\qquad \forall i,i' \in [m], \forall a_1,a_2,a'_1,a'_2 \in [d_A].
\end{aligned}
\end{equation}
For any integer $\nu$, we let $\cR^{\nu\times \nu}$ be the ring of $\nu \times \nu$ matrices where each entry is an element of $\cR$, i.e., $\cR^{\nu\times \nu}$ is the ring of $\nu\times \nu$ matrix-valued polynomials in the variables $\{\bY,\bZ\}$. In the notations established earlier, note that $\bP^{(m)}(\bZ)$ and $\bP^{(m)}(\bY)$ are both elements of $\cR^{d_A \times d_A}$.

We can express the optimization problem~\eqref{eq:ncpoly_sq_m} in the following form:
\begin{equation}
\label{eq:infL}
\begin{array}{cl}
\displaystyle\inf_{L:\cR^{d_A d_B \times d_A d_B} \to \CC} & \frac{1}{2} L\left( \bP^{(m)}(\bZ) \otimes I_B + \bP^{(m)}(\bY) \otimes I_B \right)\\
\text{s.t.} & L(\ketbra{a_1b_1}{a_2b_2} \otimes 1) = \spr{a_1b_1}{\rho_{AB}|a_2b_2}\;\; \forall a_1,a_2 \in [d_A], b_1,b_2 \in [d_B]
\end{array}
\end{equation}
where the optimization is over all linear forms $L:\cR^{d_A d_B \times d_A d_B} \to \CC$ that can be written as 
\begin{equation}
\label{eq:validL}
\begin{aligned}
L(\bP) &= \tr{ \; \sigma_{AB\cH} \; \bP(\underline{\bZ},\underline{\bY})} %\\
%&= \sum_{a_1,a_2,b_1,b_2} \tr_{\cH} \left( \sigma_{a_1b_1,a_2b_2} P_{a_1b_1,a_2b_2}(\underline{\bZ},\underline{\bY}) \right)
\end{aligned}
 \qquad \forall \bP \in \cR^{d_A d_B \times d_A d_B}
\end{equation}
where $\cH$ is a finite-dimensional Hilbert space, $\sigma_{AB\cH}$ is any state on $AB\cH$, and $\underline{\bZ},\underline{\bY}$ are operators on $\cH$.
The linear constraints on $L$ in \eqref{eq:infL} ensure that the state $\sigma_{AB\cH}$ is an extension of $\rho_{AB}$ (indeed, the notation $\ketbra{a_1b_1}{a_2b_2} \otimes 1$ is the element of $\cR^{d_A d_B \times d_A d_B}$ that has an identity in position $(a_1b_1,a_2b_2)$ and zeros elsewhere).

It is hard in general to require that $L$ has the form \eqref{eq:validL}. The $k$'th semidefinite relaxation of \eqref{eq:infL} is obtained by imposing instead the necessary condition
\begin{equation}
\label{eq:sosconstraint}
L(\bip \bip^*) \geq 0 \;\; \forall \bip \in \cR_k^{d_A d_B \times 1}
\end{equation}
where $\cR_k$ is the space of degree $k$ polynomials in $\cR$. The constraint \eqref{eq:sosconstraint} can be encoded as a positive semidefinite constraint on a matrix of size $\dim \cR_k^{d_A d_B}$. 

\subsection{Numerical implementation}

%When $k=1$, for which $\dim \cR_k^{d_A d_B} = (1 + 4md_A^2) d_A d_B$.

%Running the semidefinite program for $k = 1$. 

Figure \ref{fig:werner} shows the result of using the first level ($k=1$) of the semidefinite relaxation on the $2\times 2$ and $3\times 3$ Werner states. Recall that Werner states~\cite{werner1989quantum} in dimension $d = d_A = d_B$ are defined as $\rho_{AB} = p \frac{\Pi_{\mathrm{sym}}}{\tr(\Pi_{\mathrm{sym}})} + (1-p) \frac{\Pi_{\mathrm{asym}}}{\tr(\Pi_{\mathrm{asym}})} $, where $\Pi_{\mathrm{sym}}$ and $\Pi_{\mathrm{asym}}$ are the projectors onto the symmetric and antisymmetric subspace of $\CC^{d} \otimes \CC^{d}$ respectively. Note that when $p \geq \frac{1}{2}$, $\rho_{AB}$ is separable, so in particular the squashed entanglement is $0$.

\pgfplotstableread{
p    ub4		ub5			ub10       lb8
0.0  1.000		1.000		1.0000   0.9784
0.1  0.6419		0.6419		0.6419   0.6300
0.2  0.3979		0.3976		0.3979   0.3795
0.3  0.2127		0.2102		0.2031   0.1891
0.4  0.0810		0.0753  	0.0627   0.0472
0.5  0.0000		0.0000		0.0000  -0.0144
}\wernertbltwo

\pgfplotstableread{
p    lb10	ub4			ub5
0.0  0.7023	0.7924		0.7924
0.1  0.4521	0.6293		0.6009
0.2  0.2724	0.4972		0.4482
0.3  0.1261	0.3286		0.3047
0.4  0.0268	0.2019		0.1876
0.5  -0.02876  0.0968		0.0889
}\wernertblthree

%\subsection{Werner states}
%
\begin{figure}[ht]
\centering
\begin{tikzpicture}
\begin{axis}[width=2in,height=1.75in,at={(1.011in,0.642in)},scale only axis,
xmin=0,xmax=0.5,ymin=0,ymax=1.00,ylabel style={rotate=-90},
xlabel={$p$},ylabel={$\sq$},
legend style={legend cell align=left, align=left, draw=white!15!black},
legend entries={Upper bound $d_D=d_E=4$,Upper bound $d_D=d_E=5$,Lower bound $m=10$},
legend to name=mylegend,
axis background/.style={fill=white},title={$d_A=d_B=2$}]
\addplot [smooth,color=mycolor1, line width=1.0pt, dashed] table[x=p,y=ub4]{\wernertbltwo};% \addlegendentry{Upper bound $d_D=d_E=4$}
\addplot [smooth,color=mycolor2, line width=1.0pt, dashed] table[x=p,y=ub5]{\wernertbltwo};% \addlegendentry{Upper bound $d_D=d_E=5$}
\addplot [smooth,color=mycolor4, line width=1.0pt] table[x=p,y=lb8]{\wernertbltwo};% \addlegendentry{Lower bound $m=10$}
\end{axis}
\end{tikzpicture}%
\qquad
\begin{tikzpicture}
\begin{axis}[width=2in,height=1.75in,at={(1.011in,0.642in)},scale only axis,
xmin=0,xmax=0.5,ymin=0,ymax=1.00,ylabel style={rotate=-90},
xlabel={$p$},ylabel={$\sq$},legend style={legend cell align=left, align=left, draw=white!15!black},
axis background/.style={fill=white},title={$d_A=d_B=3$}]
\addplot [smooth,color=mycolor1, line width=1.0pt,dashed] table[x=p,y=ub4]{\wernertblthree};% \addlegendentry{Upper bound $d_D=d_E = 4$}
\addplot [smooth,color=mycolor2, line width=1.0pt,dashed] table[x=p,y=ub5]{\wernertblthree};% \addlegendentry{Upper bound $d_D=d_E = 5$}
\addplot [smooth,color=mycolor4, line width=1.0pt] table[x=p,y=lb10]{\wernertblthree};% \addlegendentry{Lower bound $m=10$}
\end{axis}
\end{tikzpicture}\\
\ref*{mylegend}
\caption{
%We consider the Werner states with $d_A = d_B$, $\rho_{AB} = p \frac{\Pi_{\mathrm{sym}}}{\tr(\Pi_{\mathrm{sym}})} + (1-p) \frac{\Pi_{\mathrm{asym}}}{\tr(\Pi_{\mathrm{asym}})} $, where $\Pi_{\mathrm{sym}}$ and $\Pi_{\mathrm{asym}}$ are the projectors onto the symmetric and antisymmetric subspace of $\CC^{d_A} \otimes \CC^{d_A}$ respectively. 
Lower bound obtained using the first level ($k=1$) of the semidefinite programming relaxation. The upper bound is obtained by a heuristic optimization for computing the infimum in~\eqref{eq:squashed_purification}. For $d = 2$, the upper and lower bounds approximately match. However, for $d=3$, there is still a large gap. We note that for the antisymmetric state (i.e., $p = 0$), explicit upper bounds on the squashed entanglement are given in~\cite{CSW12}: it gives $\frac{1}{2} \log 3 \approx 0.792$ for $d = 3$.}
\label{fig:werner}
\end{figure}
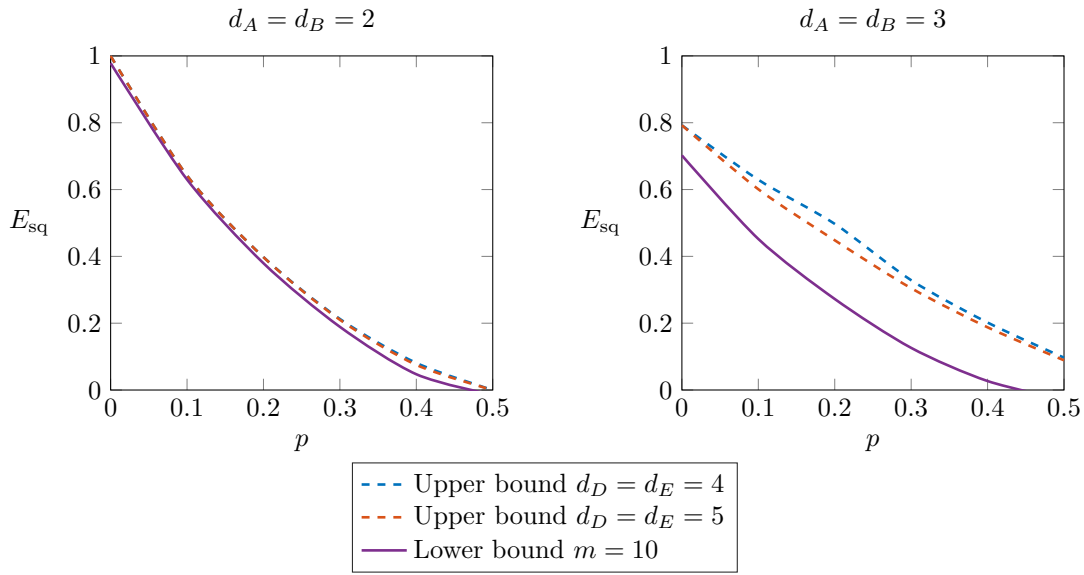

\section{Discussion}

The semidefinite programs (SDPs) that we constructed for the squashed entanglement raise a number of questions. First, does this family of SDPs converge to the squashed entanglement? We know from Proposition~\ref{prop:sqm_sq} that as $m \to \infty$, we have $\sq^{(m)}(A : B)_{\rho}$ approaches $\sq(A : B)_{\rho}$, but what remains unclear is whether the SDP relaxations for $\sq^{(m)}(A : B)_{\rho}$ converge for a fixed $m$. The difficulty comes from the fact that the general convergence results in~\cite{PNA10} give an infinite-dimensional Hilbert space but we could only prove~\eqref{eq:ncpoly_sq_m} for finite-dimensional Hilbert spaces. In other words, similar to Tsirelson's problem, the challenge here comes from comparing the commuting model to the tensor product model for the specific problem defined in~\eqref{eq:ncpoly_sq_m}. Another question is to use the specific structure of the polynomials in Theorem~\ref{prop:main_lb_sq_ent} in order to obtain finite convergence bounds for this family of SDPs, where one can hope to use techniques that were used to establish noncommutative Positivestellensatz~\cite{helton2002positive,helton2012convex}. Such a convergence could also lead to bounds on the dimension of the extension system sufficient to approximate the squashed entanglement. Note that the NP-hardness of the squashed entanglement~\cite{huang2014computing} rules out a convergence that is too fast. For example, if $\sq^{(m)}(A : B)_{\rho}$ was always equal to the SDP relaxation at level $k = 1$, we could take $m$ polynomial in the dimension and get a polyonmial-time algorithm to approximate the squashed entanglement with additive error that is inverse polynomial in the dimension, which is ruled out by~\cite{huang2014computing}. Finally, a limitation of the method we propose is that the size of the SDPs grow quickly in practice and this makes it difficult to run except in very small dimension. It would be interesting to find more efficient methods.

\section*{Acknowledgements}
We thank Peter Brown, Tim Netzer, Mizanur Rahaman and William Slofstra for useful discussions. OF acknowledges funding by the European Research Council (ERC Grant Agreement No. 851716).

\bibliographystyle{alphaabbrv}
\bibliography{big}

\end{document}